\documentclass[12pt]{amsart}

\usepackage{amsmath,amssymb,amsthm,latexsym,graphicx}
\usepackage{pstricks}
\usepackage{pstricks-add}

\newtheorem{theorem}{Theorem}[section]
\newtheorem{lemma}[theorem]{Lemma}
\newtheorem{corollary}[theorem]{Corollary}

\newtheorem{example}[theorem]{Example}

\theoremstyle{definition}
\newtheorem{definition}[theorem]{Definition}
\theoremstyle{remark}
\newtheorem{remark}[theorem]{Remark}

\newcommand{\G}{\Gamma}
\newcommand{\Gmag}{\Gamma^{\alpha_1,\ldots,\alpha_\beta}_{\mathrm{mag}}}
\newcommand{\Gmaga}{\Gamma^{\alpha}_{\mathrm{mag}}}
\newcommand{\Gm}{\Gamma_{\mathrm{mag}}}
\newcommand{\hmag}{h^{\alpha_1,\ldots,\alpha_\beta}_{\mathrm{mag}}}
\newcommand{\fmag}{f^{\alpha_1,\ldots,\alpha_\beta}_{\mathrm{mag}}}
\newcommand{\Gcut}{\Gamma_{\gamma_1,\ldots,\gamma_\beta}^{\mathrm{cut}}}
\newcommand{\hcut}{h_{\gamma_1,\ldots,\gamma_\beta}^{\mathrm{cut}}}
\newcommand{\fcut}{f_{\gamma_1,\ldots,\gamma_\beta}^{\mathrm{cut}}}
\newcommand{\Gcuta}{\Gamma_{\gamma}^{\mathrm{cut}}}
\newcommand{\Gc}{\Gamma^{\mathrm{cut}}}

\newcommand{\va}{\vec{\alpha}}
\newcommand{\vx}{\vec{x}}
\newcommand{\vz}{\vec{z}}

\newcommand{\Edges}{\mathcal{E}}

\newcommand{\Inert}{\mathcal{I}}
\newcommand{\zeros}{{\phi}}

\newcommand{\R}{\mathbb{R}}
\newcommand{\Reals}{\mathbb{R}}
\newcommand{\C}{\mathbb{C}}

\newcommand{\hg}{{\widetilde{\gamma}}}
\newcommand{\ha}{\widetilde{\alpha}}

\DeclareMathOperator{\sign}{sign}

\DeclareMathOperator{\Imag}{Im}
\renewcommand\Re{\operatorname{Re}}
\renewcommand\Im{\operatorname{Im}}

\title{Nodal count of graph eigenfunctions via magnetic perturbation}
\author{G. Berkolaiko}
\address{Department of Mathematics, Texas A\&M University, College
  Station, TX 77843-3368, USA}

\begin{document}

\begin{abstract}
  We establish a connection between the stability of an eigenvalue
  under a magnetic perturbation and the number of zeros of the
  corresponding eigenfunction.  Namely, we consider an eigenfunction
  of discrete Laplacian on a graph and count the number of edges where
  the eigenfunction changes sign (has a ``zero''). It is known that
  the $n$-th eigenfunction has $n-1+s$ such zeros, where the ``nodal
  surplus'' $s$ is an integer between 0 and the number of cycles on
  the graph.

  We then perturb the Laplacian by a weak magnetic field and view the
  $n$-th eigenvalue as a function of the perturbation.  It is shown
  that this function has a critical point at the zero field and that
  the Morse index of the critical point is equal to the nodal surplus
  $s$ of the $n$-th eigenfunction of the unperturbed graph.
\end{abstract}

\maketitle

\section{Introduction}
\label{sec:into}

Studying zeros of eigenfunctions is a question with rich history.
While experimental observations have been mentioned by Leonardo da
Vinci \cite{Leonardo}, Galileo \cite{Galileo} and Hooke \cite{Hooke},
and greatly systematized by Chladni \cite{Chladni}, the first
mathematical result is probably due to Sturm \cite{Stu_jmpa36a}.  The
Oscillation Theorem of Sturm states that the number of internal zeros
of the $n$-th eigenfunction of a Sturm-Liouville operator on an
interval is equal to $n-1$.  Equivalently, the zero of the $n$-th
eigenfunction divide the interval into $n$ parts.  In higher
dimensions, the latter equality becomes a one-sided inequality: Courant
\cite{Cou_ngwgmp23,CourantHilbert_volume1} proved that the zero curves
(surfaces) of the $n$-th eigenfunction of the Laplacian divide the
domain into at most $n$ parts (called the ``nodal domains'').

Recently, there has been a resurgence of interest in counting the
nodal domains of eigenfunctions, with many exciting conjectures and
rigorous results.  The nodal count seems to have universal features
\cite{BluGnuSmi_prl02,BogSch_prl02,NazSod_ajm09}, is conjectured to
resolve isospectrality \cite{GnuKarSmi_prl06}, and has connections to
minimal partitions of the domain
\cite{HelHofTer_aihp09,BerKucSmi_prep11}, to name but a few.  For a
selection of research articles and historical reviews, see the
dedicated volume \cite{SS07}.

On graphs, the question can be formulated regarding the signs of the
eigenfunctions of the operator
\begin{equation}
  \label{eq:H_graphs}
  H: \R^{|V|} \to \R^{|V|}, \qquad H = Q - C,
\end{equation}
where $V$ is the set of the vertices of the graph, $Q$ is an arbitrary
real diagonal matrix and $C$ is the adjacency matrix of the graph.  On
a graph, by a ``zero'' we understand an edge on which the
eigenfunction changes sign, and not the exceptional (with respect to
perturbation of $Q$) situation of an eigenfunction having a zero
entry.

The subject of sign changes and nodal domains (connected components of
the graph left after cutting the above edges) was addressed in, among
other sources, Fiedler \cite{Fie_cmj75a} who showed the analogue of
Sturm equality for \emph{tree graphs} (see also \cite{Biy_laa03}),
Davies, Gladwell, Leydold and Stadler \cite{DavGlaLeySta_laa01}, who
proved an analogue of Courant (upper) bound for the number of nodal
domains, Berkolaiko \cite{Ber_cmp08}, who proved a lower bound for
graphs with cycles and Oren \cite{Ore_jpa07}, who found a bound for
the nodal domains in terms of the chromatic number of the graph.  A
number of predictions regarding the nodal count in regular graphs
(assuming an adaptation of the random wave model) is put forward in
\cite{Elo_jpa08}.  For more information, the interested reader is
referred to the book \cite{BiyLeySta_book} and the review
\cite{BanOreSmi_pspm08}.

The study of magnetic Schr\"odinger operator on
graphs has a similarly rich history.  To give a sample, Harper
\cite{Har_ppsa55} used the tight-binding model (discrete Laplacian) to
describe the effect of the magnetic field on conduction (see also
\cite{Hof_prb76}).  In mathematical literature, discrete magnetic
Schr\"odinger operator was introduced by Lieb and Loss
\cite{LieLos_dmj93} and Sunada \cite{Sun_conf93,Sun_cm94}, and
studied, among other sources, in
\cite{Shu_cmp94,CdV_spectre,CdVToHTru_prep10} (see also
\cite{Sun_pspm08} for a review).

In this paper we present a surprising connection between the two
topics, namely the number of sign changes of $n$-th eigenfunction and the
behavior of the eigenvalue $\lambda_n$ under the perturbation of the
operator $H$ by a magnetic field.  To make a precise statement, we
need to introduce some notation.

Consider a generic eigenfunction on the graph, that is an
eigenfunction that corresponds to a simple eigenvalue and is nonzero
at the vertices.\footnote{This is the generic situation with respect to
  the perturbation of the potential $Q$.}  We denote by $\zeros_n$ the
number of \emph{sign changes} (also called \emph{sign flips}, hence
the notation $\zeros$) which are defined as the edges of the graph at
whose endpoints the eigenfunction has different signs.  Here $n$ is
the number of the eigenfunction in the sequence ordered according to
increasing eigenvalue.  The combined results of
\cite{Fie_cmj75a,Ber_cmp08,BerRazSmi_prep11} bound the number
$\zeros_n$ by
\begin{equation}
  \label{eq:nodal_bound_mu}
  n-1 \leq \zeros_n \leq n-1+\beta,
\end{equation}
where $\beta := |V|-|E|+1$ is the first Betti number (number of
independent cycles) of the graph.  Here and throughout the manuscript
we assume that the graph is connected.  We will call the
quantity
\begin{equation}
  \label{eq:nodal_surplus_def}
  \sigma_n = \zeros_n - (n-1), \qquad 0 \leq \sigma_n \leq \beta
\end{equation}
the \emph{nodal surplus}.  This is the extra number of sign changes that an
eigenfunction has due to the graph's non-trivial topology.

Magnetic field on discrete graphs has been introduced in, among other
sources, \cite{LieLos_dmj93,Sun_cm94,CdV_spectre}.  Up to unitary
equivalence, it can be specified using $\beta$ phases $\va =
(\alpha_j)_{j=1}^\beta \in (-\pi,\pi]^\beta$.  We consider the
eigenvalues of the graph as functions of the parameters $\va$.  The
zero phases, $\va = 0$, correspond to the graph $\G$ without the
magnetic field.  We are now ready to formulate our main result.

\begin{theorem}
  \label{thm:main}
  The point $\vec{\alpha} = 0$ is the critical point of the function
  $\lambda_n(\vec{\alpha})$.  If $\lambda_n(0)$, the $n$-th
  eigenvalue of the non-magnetic operator on $\G$, is simple and the
  corresponding eigenfunction has no zero entries, its nodal surplus
  $\sigma_n$ is equal to the Morse index --- the number of negative
  eigenvalues of the Hessian --- of $\lambda_n(\vec{\alpha})$ at the
  critical point $\vec{\alpha}=0$.
\end{theorem}

An immediate consequence of this theorem is the following.

\begin{corollary}
  The non-degenerate $n$-the eigenvalue of the discrete Schr\"odinger
  operator is stable with respect to magnetic perturbation of the
  operator if and only if the corresponding eigenfunction has exactly
  $n-1$ sign changes.  

  By ``non-degenerate'' we understand a simple eigenvalue whose
  eigenfunction does not vanish on vertices and by ``stability'' we
  mean that the eigenvalue has a local minimum at zero magnetic
  field.
\end{corollary}

Other possible consequences of our result and links to several other
questions are discussed in Section~\ref{sec:discuss}.  The rest of the
paper is structured as follows.  In Section~\ref{sec:mag_ham} we
provide detailed definitions.  Section~\ref{sec:duality} is devoted to
a duality between the magnetic perturbation and a certain perturbation
to the potential, coupled with removal of edges.  This leads to an
alternative proof of the result in the case $\beta=1$
(subsection~\ref{sec:duality_proof}) which, although unnecessary for
the general proof, provides us with some important insights.
Section~\ref{sec:tools} collects the tools necessary for the proof of
Theorem~\ref{thm:main}, while Section~\ref{sec:proof_main} contains
the proof itself.

\section{Magnetic Hamiltonian on discrete graphs}
\label{sec:mag_ham}

Let $\G = (V, E)$ be a simple finite graph with the vertex set $V$
and the edge set $E$.  We define the Schr\"odinger operator with the
potential $q : V\to\Reals$ by
\begin{equation}
  \label{eq:discr_schrod}
  H: \Reals^{|V|} \to \Reals^{|V|}, \qquad
  (H\psi)_u = - \sum_{v\sim u} \psi_v + q_u\psi_u,
\end{equation}
that is the matrix $H$ is
\begin{equation}
  \label{eq:matrix_H}
  H = Q - C,
\end{equation}
where $Q$ is the diagonal matrix of site potentials $q_u$ and $C$ is
the adjacency matrix of the graph.  It is perhaps more usual (and
physically motivated) to represent the Hamiltonian as $H = Q + L$,
where the Laplacian $L$ is given by $L = D - C$ with $D$ being the
diagonal matrix of vertex degrees.  But since we will not be imposing
any restrictions on the potential $Q$, we absorb the matrix $D$ into
$Q$.

The operator $H$ has $|V|$ eigenvalues, which we number in increasing
order,
\begin{displaymath}
  \lambda_1 \leq \lambda_2 \leq \ldots \leq \lambda_{|V|}.
\end{displaymath}

We define the magnetic Hamiltonian (magnetic Schr\"o\-dinger operator)
on discrete graphs as
\begin{equation}
  \label{eq:discr_mag_schrod}
  (H\psi)_u = - \sum_{v\sim u} e^{iA_{v,u}}\psi_v + q_u\psi_u,
\end{equation}
with the convention that $A_{v,u} = -A_{u,v}$, which makes $H$
self-adjoint.  For further details, the reader should consult
\cite{LieLos_dmj93,Sun_cm94,CdV_spectre,CdVToHTru_prep10}.

A sequence of directed edges $C = [u_1,u_2,\ldots,u_n]$ is called a
cycle if the terminus of edge $u_j$ coincides with the origin of the
edge $u_{j+1}$ for all $j$ ($u_{n+1}$ is understood as $u_1$).  The
flux through the cycle $C$ by
\begin{equation}
  \label{eq:flux_through_C}
  \Phi_C = A_{u_1,u_2} + \ldots + A_{u_{n-1},u_n} + A_{u_n,u_1}
  \mod 2\pi
\end{equation}
Two operators which have the same flux through every cycle $C$ are
unitarily equivalent (by a gauge transformation).  Therefore, the
effect of the magnetic field on the spectrum is fully determined by
$\beta$ fluxes through a chosen set of basis cycles of the cycle
space.  We denote them by $\alpha_1,\ldots, \alpha_\beta$ and consider
the $n$-th eigenvalue of the graph as a function of $\vec{\alpha}$.

More precisely, fix an arbitrary spanning tree of the graph and let
$S$ be the set of edges that do not belong to the chosen tree.
Obviously, $S$ contains exactly $\beta$ edges.

\begin{lemma}
  Any magnetic Schr\"odinger operator on the graph $\G$ is unitarily
  equivalent to one of the operators of the type
  \begin{equation}
    \label{eq:S_mag_schrod}
    H_{u,v} =
    \begin{cases}
      V_{u}, & u=v,\\
      -1, & (u,v) \in \Edges\setminus S,\\
      -e^{\pm i \alpha_s}, & (u,v) = s \in S,
    \end{cases}
  \end{equation}
  where the sign of the phase is plus if $u<v$ and minus if $u>v$.
\end{lemma}

\begin{example}
  Consider the triangle graph --- a graph with three vertices and tree
  edges connecting them.  One of the equivalent forms of the magnetic
  Hamiltonian for this graph is
  \begin{equation*}
    H(\Gmaga) =
    \begin{pmatrix}
      q_1 & -e^{i \alpha} & -1\\
      -e^{-i \alpha} & q_2 & -1\\
      -1 & -1 & q_3
    \end{pmatrix}.
  \end{equation*}
  The spectrum of $H(\Gmaga)$ as a function of
  $\alpha\in(-\pi,\pi]$ is shown in Fig.~\ref{fig:cut_mag_duality}.
\end{example}

\begin{figure}[t]
  \centering
  \includegraphics{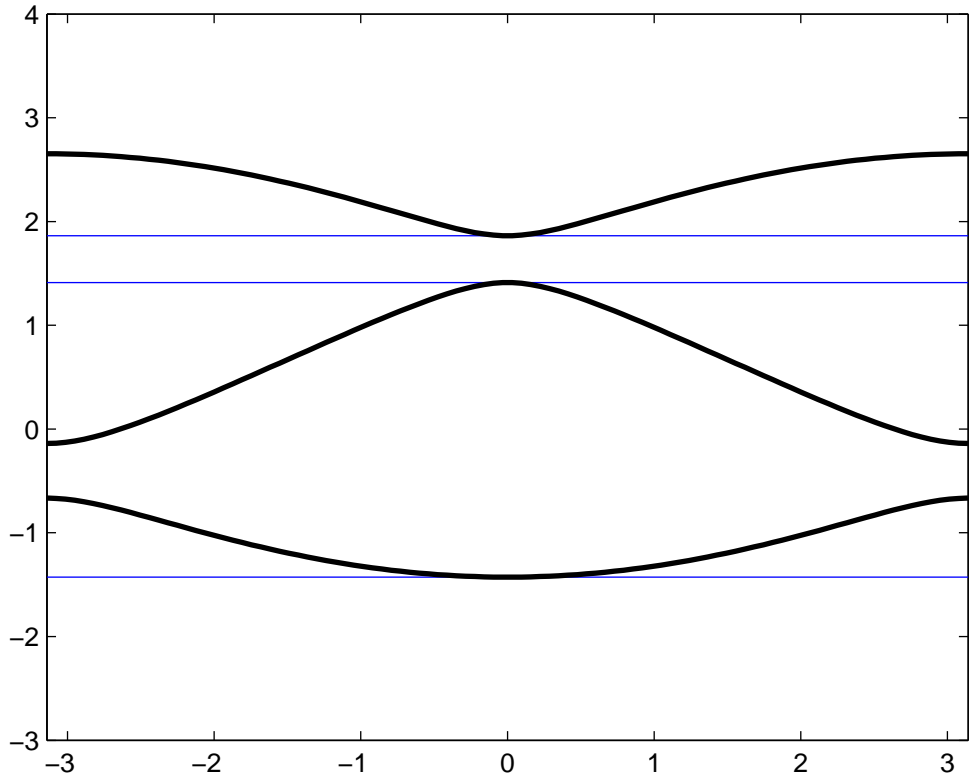}
  \caption{The eigenvalues of the triangle graph as functions of a
    magnetic phase $\alpha$ (bold lines) and the eigenvalues of the
    unperturbed graph (horizontal lines).} 
  \label{fig:mag_eig}
\end{figure}

\section{A duality between a magnetic phase and a cut}
\label{sec:duality}

In this section we explore a simple result which shows a connection
between two types of perturbations of the operator $H$ that will be
used to prove the main theorem.  It illustrates the duality between
the perturbation of a discrete Schr\"odinger operator by a magnetic
phase on a cycle and the operation of removing (``cutting'') an edge
that lies on the cycle.  The latter operation was used to prove the
lower bound on the number of nodal domains in \cite{Ber_cmp08} and to
study partitions on discrete graphs in \cite{BerRazSmi_prep11}.

\subsection{Tools used}
The result of this section (Theorem \ref{thm:cut_mag_interlacing}
below) is based on the following version of the Weyl's inequality of
linear algebra that can be obtained using the variational
characterization of the eigenvalues (see
\cite[Chap. 4]{HornJohnson_matrix} for similar results).

\begin{theorem}
  \label{thm:interlacing_LA}
  Let $A$ be a self-adjoint matrix and $B$ be a rank-one positive
  semidefinite self-adjoint matrix.  Then
  \begin{equation}
    \label{eq:interlacing_LA}
    \lambda_n(A-B) \leq \lambda_n(A) \leq \lambda_{n+1}(A-B),
  \end{equation}
  where $\lambda_n$ is the $n$-th eigenvalue, numbered in increasing
  order, of the corresponding matrix.  Moreover, the inequalities are
  strict if and only if $\lambda_n(A)$ is simple and its eigenvector
  is not in the null-space of $B$.

  Similarly, when $B$ is negative-definite, we have
  \begin{equation}
    \label{eq:interlacing_LA_neg}
    \lambda_{n-1}(A-B) \leq \lambda_n(A) \leq \lambda_{n}(A-B),
  \end{equation}
  with an analogous condition for strict inequalities.
\end{theorem}

Another useful result is the first term in the perturbation expansion
of a parameter-dependent eigenvalue.  Let $A(x)$ be a Hermitian
matrix-valued analytic function of $x$.  Let $\lambda(x)$ be an
eigenvalue of the matrix $A$ that is simple in a neighborhood of a
point $x_0$.  We know from standard perturbation theory
\cite{Kato_book} that $\lambda(x)$ is an analytic function.  Denote by
$u(x)$ the normalized eigenvector of corresponding to the eigenvalue
$\lambda$ and by $v$ all other normalized eigenvectors (in a slight
abuse of notation).  Then we have the following formula for the
derivative of $\lambda$ evaluated at the point $x=x_0$.
\begin{equation}
  \label{eq:pert_first_deriv}
  \frac{\partial}{\partial x} \lambda = \left\langle u, \frac{\partial
      A}{\partial x} u \right\rangle.
\end{equation}

\subsection{Two operations on a graph}
\label{sec:two_operations}
Let $\lambda_n$ be a simple eigenvalue and the corresponding
eigenfunction $f$ be non-zero on vertices.  Let $(u_1,u_2)$ be an edge
that belongs to one of cycles of the graph.  We allow the graph to
have magnetic phases on some edges, but assume that there is no phase
on the edge $(u_1,u_2)$.  Then the operator $H=Q-C$ has the following
subblock corresponding to vertices $u_1$ and $u_2$,
\begin{equation}
  \label{eq:subblock_orig}
  H(\G)_{[u_1, u_2]} = \left(
    \begin{array}{cc}
      q_{u_1} & -1 \\
      -1 & q_{u_2}
    \end{array}
  \right).
\end{equation}

We consider two modifications of the original graph.  The first
modification of the graph is a cut: we remove the edge $(u_1,u_2)$ and
change the potential at sites $u_1$ and $u_2$.  Namely, we change the
$[u_1, u_2]$ subblock to
\begin{equation}
  \label{eq:subblock_cut}
  H(\Gcuta)_{[u_1, u_2]} = \left(
    \begin{array}{cc}
      q_{u_1}-\gamma & 0 \\
      0 & q_{u_2}-1/\gamma
    \end{array}
  \right),
\end{equation}
and leave the rest of the matrix $H$ intact.  We denote this
modification by $H(\Gcuta)$.  Note that this modification is a
rank-one perturbation of the original operator $H(\G)$.  Let $B$ be
the perturbation such that $H(\Gcuta) = H(\G) - B^c$.  Namely, the
matrix $B^c$ has the $[u_1,u_2]$ subblock
\begin{equation}
  \label{eq:subblock_cut_pert}
  B^c_{[u_1, u_2]} =
  \left(
    \begin{array}{cc}
      \gamma & -1 \\
      -1 & 1/\gamma
    \end{array}
  \right),
\end{equation}
and the rest of the elements are zero.  Then $B^c$ is positive-definite
if $\gamma>0$ and negative-definite if $\gamma<0$.  Note that the
cases $\gamma=\infty$ and $\gamma=0$ can also be given meaning of
removing (or imposing the Dirichlet condition at) the vertex $u_1$ or
the vertex $u_2$ correspondingly.  However, we will not dwell on this
issue and exclude these cases from our consideration.

Notably, if $f$ is an eigenfunction of $H(\G)$ and $\gamma =
f_{u_2} / f_{u_1} \in \Reals$, then $f$ is also an eigenfunction of
$H(\Gcuta)$.  Equivalently, $f$ is in the null-space of the
perturbation $B^c$.

The second modification of the original graph is the introduction of a
magnetic phase on the edge $(u_1,u_2)$.  The $[u_1, u_2]$ subblock of
the new operator $H(\Gmaga)$ is 
\begin{equation}
  \label{eq:subblock_mag}
  H(\Gmaga)_{[u_1, u_2]} = \left(
    \begin{array}{cc}
      q_{u_1} & -e^{i\alpha} \\
      -e^{-i\alpha} & q_{u_2}
    \end{array}
  \right),
\end{equation}
while other entries coincide with those of $H(\G)$.  Note that
$H(\Gmaga)$ is \emph{not} a rank-one perturbation of $H(\G)$.
However, it is a rank-one perturbation of the cut graph $H(\Gcuta)$
for any values of $\alpha$ and $\gamma$.  Namely, $H(\Gcuta) =
H(\Gmaga) - B^{mc}$, where
\begin{equation}
  \label{eq:subblock_cut_mag_pert}
  B^{mc}_{[u_1, u_2]} =
  \left(
    \begin{array}{cc}
      \gamma & -e^{i\alpha_j} \\
      -e^{-i\alpha_j} & 1/\gamma
    \end{array}
  \right),
\end{equation}
and all other entries of $B^{mc}$ are zero.  Also, the spectrum of
$H(\Gmaga)$ and $H(\G)$ coincide when $\alpha=0$ since the
operators coincide.

\subsection{A duality between the two operations}
\label{sec:duality_proof}

We now want to apply Theorem~\ref{thm:interlacing_LA} to the spectra
of $\G$, $\Gcuta$ and $\Gmaga$.  However, we must take care to
distinguish the two cases that correspond to equations
\eqref{eq:interlacing_LA} and \eqref{eq:interlacing_LA_neg}
($\gamma>0$ and $\gamma<0$ correspondingly).  


\begin{definition}
  \label{def:numbering}
  The eigenvalues of $\G$, $\Gcuta$ and $\Gmaga$ will be
  numbered in increasing order starting from 1.  When we happen to
  index ``nonexistent'' eigenvalues we use the following convention:
  \begin{equation*}
    \lambda_j(\G) = 
    \begin{cases}
      -\infty, & j<1,  \\
      \infty, &j>n.
    \end{cases}
 \end{equation*}
\end{definition}

\begin{theorem}
  \label{thm:cut_mag_interlacing}
  Let $p(\gamma)$ be 1 if $\gamma<0$ and $0$ otherwise.  Then the following
  inequalities hold
  \begin{equation}
    \label{eq:cut_mag_interlacing}
    \lambda_{n-p(\gamma)}(\Gcuta) \leq \lambda_n(\Gmaga)
    \leq \lambda_{n-p(\gamma)+1}(\Gcuta),
  \end{equation}
  for all values of $\alpha$ and $\gamma$.  
  Furthermore, for any fixed $n$
  \begin{equation}
    \label{eq:extrema_1}
    \max_\gamma \lambda_{n-p(\gamma)}(\Gcuta) 
    = \min_\alpha \lambda_n(\Gmaga)
  \end{equation}
  and
  \begin{equation}
    \label{eq:extrema_2}
    \max_\alpha \lambda_n(\Gmaga)
    = \min_\gamma \lambda_{n-p(\gamma)+1}(\Gcuta). 
  \end{equation}
  Finally, if there are no magnetic phases on the graph $\G$ (i.e. all
  entries of $H(\G)$ are real), then one of the extrema
  \eqref{eq:extrema_1} or \eqref{eq:extrema_2} is equal to
  $\lambda_n(\G) = \lambda_n(\Gm^{\alpha=0})$, while the other is equal
  to $\lambda_n(\hat{\G}) := \lambda_n(\Gm^{\alpha=\pi})$.
\end{theorem}

\begin{remark}
  Note that at this point we don't know which extremum,
  \eqref{eq:extrema_1} or \eqref{eq:extrema_2}, is equal to
  $\lambda_n(\G)$.  This information is related to the nodal surplus.
  We have also defined yet another modification of the graph $\G$, the
  graph $\hat{\G}$ whose adjacency matrix has $-1$ in place of $1$ for
  the entries $C_{u_1,u_2}$ and $C_{u_1,u_2}$.
\end{remark}

\begin{remark}
  Let $\overline{\R} = \R \cup \{-\infty, \infty\}$ be the extended
  real line and $\widehat{\R} = \overline{\R} / [-\infty = \infty]$ be
  its projective (``wrapped'') version.  The eigenvalue
  $\lambda_{n-p(\gamma)}(\Gc_\gamma)$ is then a continuous function
  of $\gamma$, considered as a function from $\widehat{\R}$ to
  $\overline{\R}$.  See Figure~\ref{fig:cut_mag_duality} for an
  example.  Note that according to our definitions,
  $\lambda_{n-p(\gamma)}(\Gc_\gamma) = -\infty$ for $n=1$ and
  $\gamma<0$.
\end{remark}

\begin{figure}[t]
 \centering
 \includegraphics{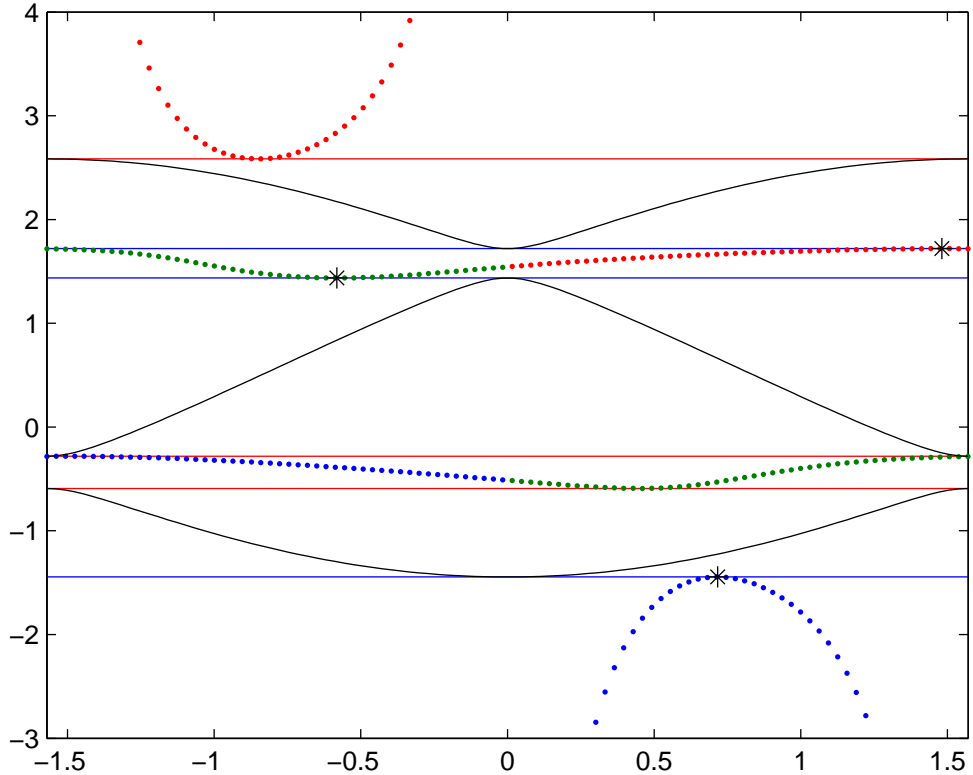}
 \caption{The duality between a magnetic field on one side and cut
   edge with added potential on the other.  The graph is a triangle.
   The black curves correspond to the eigenvalues as functions of the
   magnetic phase.  Colored symbols correspond to varying potential
   after cutting the edge.  The $x$-axis ranges from $-\pi/2$ to $\pi/2$
   with the magnetic phase taken as $\alpha = 2x$ and the potential
   parameter $\gamma=\tan(x)$.  The horizontal blue lines are
   eigenvalues of the true graph, while horizontal red lines are
   eigenvalues with magnetic phase $\pi$.}
 \label{fig:cut_mag_duality}
\end{figure}

\begin{proof}[Proof of Theorem~\ref{thm:cut_mag_interlacing}]
  The inequalities follow directly from
  Theorem~\ref{thm:interlacing_LA}, since for any $\alpha$ the graph
  $\Gmaga$ is a rank-one perturbation of $\Gcuta$.  Whether it
  is positive- or negative-definite depends on the sign of $\gamma$,
  and results in the shift by $p$.

 The properties of the extrema we can get as follows.  First of all,
  observe that if $\max \lambda_{n-p}(\Gcuta) = \min
  \lambda_{n-p+1}(\Gcuta)$, then $\lambda_n(\Gm^\alpha)$ is
  constant and equal to the common value of $\lambda_{n-p}(\Gcuta)$ and
  $\lambda_{n-p+1}(\Gcuta)$.

  Let now $\max \lambda_{n-p}(\Gcuta) < \min
  \lambda_{n-p+1}(\Gcuta)$.  The eigenvalues of a one-parameter
  family can always be represented as a set of analytic functions
  (that can intersect).  Let $\lambda'(\Gcuta)$ be the analytic
  function that achieves the maximum $\max \lambda_{n-p}(\Gcuta)$
  and $f$ be the corresponding eigenfunction.  We will differentiate
  $\lambda'(\Gcuta)$ using equation~\eqref{eq:pert_first_deriv}.
  At the maximum point $\gamma = \hg$ we have (see
  equation~\eqref{eq:subblock_cut}),
  \begin{equation}
    \label{eq:deriv_at_max}
    0 = \frac{d \lambda'}{d\gamma} 
    = \left\langle f, \frac{d B^c}{d\gamma} f\right\rangle
    = -|f_{u_1}|^2 + |f_{u_2}|^2 / \hg^2.
  \end{equation}
  From here it follows that
  \begin{equation}
    \label{eq:gamma}
    \hg = \pm \frac{|f_{u_2}|}{|f_{u_1}|} \qquad \mbox{or, equivalently,} \qquad
    \left| \hg f_{u_1}/f_{u_2} \right| = 1.
  \end{equation}
  Let $\ha$ be the solution of $e^{i\alpha} = \hg
  f_{u_1}/f_{u_2}$.  Direct calculation shows that the eigenfunction
  $f$ is in the null-space of the perturbation $B^{mc}$ of
  \eqref{eq:subblock_cut_mag_pert} with $\alpha=\ha$ and
  therefore $f$ is both in the spectrum of $\Gc_{\hg}$ and in the
  spectrum of $\Gm^{\ha}$, implying equation~\eqref{eq:extrema_1}
  follows.  Proof of equation~\eqref{eq:extrema_2} is completely
  analogous.

  Note that we could instead differentiate the eigenvalue of
  $\Gmaga$, leading to the condition
  \begin{equation}
    \label{eq:alpha_deriv_zero}
    f_{u_2} \overline{f_{u_1}} e^{i\ha} \in \Reals,
  \end{equation}
  instead of equation~\eqref{eq:deriv_at_max}.  One then sets
  $\hg = e^{i\ha}f_{u_2}/f_{u_1} \in \Reals$ to the same effect.

  Finally, when the matrix $H(\G)$ is real, the eigenfunctions of
  $\Gcuta$, $\Gm^{\alpha=0} $ and $\Gm^{\alpha=\pi}$ are real-valued.
  When $\alpha=0$ we can verify directly that the eigenfunction $f$ of
  $\Gc_{\alpha=0}$ is also an eigenfunction of $\Gc_{\hg}$ by setting
  $\hg = f_{u_2}/f_{u_1}$.  When $\alpha=\pi$, we also set $\gamma =
  f_{u_2}/f_{u_1}$ and do the same.
\end{proof}

Theorem~\ref{thm:cut_mag_interlacing} highlight a sort of duality
between the two modifications of the graph $\G$.  The spectra of the
graphs with a magnetic phase form bands (as the phase is varied) while
the spectra of the graphs with the cut fill the gaps between this
bands.  Minimums of one correspond to maximums of the other and in
half of the cases correspond to eigenvalues of the original graph.

We now explain how the $\beta=1$ case of Theorem~\ref{thm:main}
follows from Theorem~\ref{thm:cut_mag_interlacing}.  While for general
$\beta$ the proof is significantly different (it bypasses the interlacing
inequalities and goes straight to the quadratic form), some key
features are the same as in this simple case.

Starting with the eigenvalue $\lambda_n$ of $\G$ and the corresponding
eigenfunction $f$, we cut an edge on the only cycle of $\G$ to obtain
a family of trees $\Gcuta$.  For $\gamma = \hg = f_{u_2} /
f_{u_1}$, we have either
\begin{equation*}
  \max_\gamma \lambda_{n-p(\gamma)}(\Gcuta) 
  = \lambda_{n-p(\hg)}(\Gc_\hg) = \lambda_n(\G) 
  = \min_\alpha \lambda_n(\Gmaga),
\end{equation*}
or 
\begin{equation*}
  \max_\alpha \lambda_n(\Gmaga)
  = \lambda_n(\G)
  = \lambda_{n-p(\hg)+1}(\Gc_\hg)
  = \min_\gamma \lambda_{n-p(\gamma)+1}(\Gcuta).
\end{equation*}
In the first case, according to Fiedler theorem (equation
\eqref{eq:nodal_bound_mu} with $\beta=0$), the function $f$ has
$n-p(\hg)-1$ sign changes \emph{with respect to} the tree $\Gcuta$.
Adding back the removed edge $(u_1,u_2)$ adds another sign change if $\hg <
0$ and doesn't change the number of sign changes otherwise.  In other words,
it adds $p(\hg)$ sign changes.  Thus, with respect to $\G$, the function $f$
has $n-1$ sign change and $\sigma_n=0$.
In the second case, we similarly conclude that $f$ has $n-p(\hg)$
sign changes with respect to $\Gcuta$ and $n$ sign changes with respect to $\G$.
The nodal surplus is $\sigma_n=1$.  

On the other hand, in the first case $\lambda_n(\G)$ is a minimum of
$\lambda_n(\Gmaga)$ (Morse index 0), while in the second it is a
maximum of $\lambda_n(\Gmaga)$ (Morse index 1), which shows that
the Morse index coincides with $\sigma_n$ in the case $\beta=1$.

\section{Tools of the main proof}
\label{sec:tools}

In this section we collect some basic facts that will be repeatedly
used in the proof of Theorem~\ref{thm:main}.

\subsection{Critical points of the quadratic form}

\begin{definition}
  Let $F: \R^d \to \R$ be a twice differentiable function.  If $c$ is
  a critical point (i.e. $\nabla f(c) = 0$), the \emph{inertia} of $c$
  is the triple $(n_-, n_0, n_+)$ that counts the number of negative,
  zero and positive eigenvalues correspondingly of the Hessian (the
  matrix of second derivatives) at the point $c$.  The number $n_-$ is
  called the \emph{Morse index} (or simply \emph{index}).
\end{definition}

The next theorem is a reminder that the eigenvectors of a self-adjoint
matrix are critical points of the quadratic form on the unit sphere.

\begin{theorem}
  \label{thm:index_of_eigenvector}
  Let $A$ be an $d\times d$ real symmetric matrix and $h(x) =
  \langle x, Ax \rangle$, $x\in\R^d$, be the associated quadratic form.
  Then the (real) eigenvectors of the matrix $A$ are critical points
  of the function $h(x)$ on the unit sphere $\|x\|=1$.  

  Let $\lambda_n$ be the $n$-th eigenvalue of $A$ and $f^{(n)}$ be the
  corresponding normalized eigenfunction.  Define
  \begin{equation}
    \label{eq:n_pm}
    n_{-} = \#\{\lambda_m < \lambda_n\},\qquad
    n_{0} = \#\{\lambda_m = \lambda_n,\ m\neq n\},\qquad
    n_{+} = \#\{\lambda_m > \lambda_n\},
  \end{equation}
  with $n_{-}+n_0+n_+ = d-1$.  Then the inertia of the critical
  point $x = f^{(n)}$ is $(n_-, n_0, n_+)$.  In particular, if
  $\lambda_n$ is a simple eigenvalue, the inertia is $(n-1, 0, d-n)$.
\end{theorem}

\begin{remark}
  \label{rem:value_at_crit}
  The value of the quadratic form $h$ at the critical point $f^{(n)}$
  is $\lambda_n$.
\end{remark}

\begin{proof}
  The idea is intuitively clear: $n_{-}$ --- which is the Morse index
  --- counts the number of directions in which the quadratic form
  decreases relative to the value at $x=f^{(n)}$.  These directions
  are the eigenvectors corresponding to the eigenvalues that are less
  than $\lambda_n$.  Similar characterizations are valid for $n_0$ and
  $n_+$. 

  We note that by Sylvester law of inertia, the inertia is invariant
  under the change of variables.  Making the orthogonal change of
  coordinates to the eigen-basis of the matrix $A$, the quadratic form
  $h(a)$ becomes
  \begin{equation*}
    h(a) = \lambda_1 a_1^2 + \lambda_2 a_2^2 + \ldots + \lambda_d a_d^2,
  \end{equation*}
  while the sphere is given by the equations
  \begin{equation*}
    a_1^2 + a_2^2 + \ldots + a_d^2 = 1.
  \end{equation*}
  Thus, on the sphere, the quadratic
  form in terms of variables $a_1,\ldots, a_{n-1}, a_{n+1}, \ldots,
  a_d$ is given by
  \begin{equation*}
    h = \sum_{j\neq n} (\lambda_j-\lambda_n) a_j^2,
  \end{equation*}
  and the Hessian is a diagonal matrix with $\lambda_j-\lambda_n$,
  $j=1,\ldots, d$, $j\neq n$.  The statement of the Theorem follows
  immediately.
\end{proof}

In the more general case when the matrix is Hermitian we should
consider the quadratic form on the space $\C^d$.  However, we can also
consider it on the real space of double dimension.

\begin{theorem}
  \label{thm_index_of_eigenvector_complex}
  Let $A$ be an $d\times d$ Hermitian matrix and $h(z) =
  \langle z, Az \rangle$, $z\in\C^d$, be the associated quadratic
  form.  Consider $h$ as a function of $2d$ real variables $(x,y)$,
  where $x\in \R^d$, $y\in\R^d$ and $z=x+iy$.
  Then the eigenvectors of the matrix $A$ are critical points
  of the function $h(x+iy)$ on the unit sphere $\|x+iy\|=1$.  

  Let $\lambda_n$ be the $n$-th eigenvalue of $A$ and $f^{(n)}$ be a
  corresponding eigenfunction.  Then the inertia of the critical
  point $x+iy = f^{(n)}$ with respect to the real space $\R^{2d}$ is
  $(2n_-, 2n_0+1, 2n_+)$, where $n_-$, $n_0$ and $n_+$ are defined by
  \eqref{eq:n_pm}.  In particular, if $\lambda_n$ is a simple
  eigenvalue, the inertia is $(2n-2, 1, 2d-2n)$.
\end{theorem}

\begin{proof}
  We adopt the convention that in the scalar product $\langle z, w
  \rangle$ the conjugation is applied to the first argument.  Also,
  when mapping a $z\in\C^d$ to $\R^{2d}$, we first list the real parts
  of the components of $z$ and then the imaginary parts.  With this
  conventions, it's an easy exercise to show that the quadratic form
  $h(z)$ in variables $(x,y)$ corresponds to the matrix
  \begin{equation}
    \label{eq:quad_form_in_double}
    B =
    \begin{pmatrix}
      \Re A & -\Im A\\
      \Im A & \Re A
    \end{pmatrix}.
  \end{equation}
  Note that because $A$ is Hermitian, the matrix $\Im A$ is
  scew-symmetric and, therefore, the matrix $B$ is real symmetric.
  Every eigenvector $z$ of the matrix $A$ corresponds to 2 real
  eigenvectors of $B$ with the same eigenvalue, namely $(\Re z, \Im
  z)$ and $(\Re(iz), \Im(iz)) = (-\Im z, \Re z)$.  These eigenvectors
  are orthogonal to each other, and to other similarly obtained
  eigenvectors.  We therefore conclude that the spectrum of $B$ is the
  same as the spectrum of $A$ with all multiplicities doubled.  The
  statement now follows from Theorem~\ref{thm:index_of_eigenvector}.
\end{proof}

\subsection{Reduction to critical manifold}

The tool introduced in this section is a simple idea already used in
\cite{Banetal_cmpXX,BerKucSmi_prep11,BerRazSmi_prep11}.  If we have a
function $f(x_1,\ldots,x_n)$ with a critical point $c$ then under some
general conditions there is a $(n-1)$-dimensional manifold around the
point $c$ on which the local minimum of $f$ is achieved when we vary
the variable $x_1$ and keep the others fixed.  Then the Morse index of
$f$ restricted to this manifold is the same as the Morse index of the
unrestricted function.  On the other hand, if the manifold is the
locus of local maxima with respect to the variable $x_1$, the Morse
index on the manifold is one less than the unrestricted Morse index.
The following theorem is a simple generalization of this idea.  The
proof is adapted from \cite{BerKucSmi_prep11}.

\begin{theorem}[Reduction Theorem]
  \label{thm:reduction}
  Let $X=Y \oplus Y^\perp$. Let $f:X\to \R$ be a smooth functional
  such that $(0,0)\in X$ is its non-degenerate critical point with
  inertia $\Inert_X$.  Further, let, for every $y\in Y$ locally around
  $0$, the functional $f(y,y')$ as a function of $y'$ has a critical
  point at $y'=0$ with inertia $\Inert_{Y^\perp}$, that (locally) does
  not depend on $y$.  Then the Hessian of $f$ is reduced by the
  decomposition $X=Y \oplus Y^\perp$ and the inertia of $f$ with
  respect to the space $Y$ is
  \begin{equation}
    \label{eq:reduced_sig}
    \Inert_Y = \Inert_X - \Inert_{Y^\perp}.
  \end{equation}
\end{theorem}

\begin{proof}
  We calculate the mixed derivative of $f$ with respect to one
  variable from $Y$ and the other from $Y^\perp$.  In a slight abuse
  of notation we denote these variables simply by $y$ and $y'$.  We
  have
  \begin{equation}
    \label{eq:mixed_deriv}
    \frac{\partial^2 f}{\partial y\, \partial y'}(0,0) 
    = \left. \frac{\partial}{\partial y}\left[\frac{\partial
          f}{\partial y'}(y,0) \right] \right|_{y=0} = 0,
 \end{equation}
  since $y'=0$ is the critical point of $f(y,y')$ as a function of
  $y'$ for every $y$.  Thus the Hessian of $f$ has a block-diagonal
  form with two blocks that correspond to $Y$ and $Y^\perp$.  The
  spectrum of the Hessian is the union of the spectra of the blocks
  and the inertia is the sum of the inertias of the blocks,
  \begin{equation*}
    \Inert_X = \Inert_Y + \Inert_{Y^\perp}.
  \end{equation*}
  Equation~\eqref{eq:reduced_sig} follows immediately.
\end{proof}

The restriction of the function $f$ to the subspace $Y$ will be called
the \emph{reduction} of $f$ with inertia $\Inert_{Y^\perp}$.  More
often than not, we will be concerned with only the index of a critical
point (the entry $n_{-}$ of the inertia).  In this case we will say
``reduction with index $m$''.  The space $Y$ is the locus of critical
points of $f$ with respect to the variable $y'$ and we will refer to
it as a \emph{critical manifold}.

\begin{remark}
  \label{rem:crit_manifold}
  Theorem~\ref{thm:reduction} can be simply extended to the case when,
  for every fixed $y$, the critical point with respect to $y'$ is
  located at $y'=q(y)$ (rather than $y'=0$).  The function $q(y)$
  defines the critical manifold $\mathcal{Q} = (y,q(y))$.  If $q(y)$
  is a smooth function of $y$ and $q(0)=0$, the change of variables
  \begin{equation*}
    y \mapsto y, \qquad y' \mapsto y' - q(y)
  \end{equation*}
  is non-degenerate (its Jacobian is a triangular matrix with 1s on
  the diagonal) and makes $f$ satisfy the assumptions of
  Theorem~\ref{thm:reduction}.  By Sylvester law of inertia, the
  conclusion of the Theorem is invariant under the change of
  variables, that is inertia at point $0$ of $f\big|_{\mathcal{Q}}$ is
  \begin{equation*}
    \Inert_{\mathcal{Q}} = \Inert_X - \Inert_{Y^\perp}.
  \end{equation*}
\end{remark}

\section{Proof of the main theorem}
\label{sec:proof_main}

We prove the main result in three steps.  First we show by an explicit
computation that the point $0$ is the critical point of the function
$\lambda_n(\alpha)$, where $\alpha = (\alpha_1,\ldots,\alpha_\beta)
\in (-\pi, \pi]^\beta$ are the magnetic phases.

Then we fix an eigenpair $\lambda = \lambda_n(\G)$ and $f$.  We cut
$\beta$ edges of the graph turning it into a tree $T$, but modifying
the potentials so that the eigenfunction $f$ is also an eigenfunction
of the tree $T$.  It now corresponds to an eigenvalue number $m$, that
is $\lambda_m(T)=\lambda$.  Considering the eigenvalue $\lambda_m(T)$
as a function of the potentials, we find its inertia.  This is done
by considering the inertia of the corresponding quadratic form in
the real space.  The result of this step is related to the results on
critical partitions, \cite{BerRazSmi_prep11}.

Finally, we relate the inertia of the function $\lambda_m(T)$ to the
inertia of the function $\lambda_n(\alpha)$ at the corresponding
critical points.  This is done by relating the inertias of the
quadratic forms, but now in the complex space represented as a real
space of double dimension.

We recall that $S$ is a set of $\beta$ edges whose removal turns the
graph $\G$ into a tree.  By $\Gmag$ we denote the graph obtained
from $\G$ by introducing magnetic phases $\alpha_1,\ldots,
\alpha_\beta$ on the edges from the set $S$.  Similarly, by $\Gcut$ we
denote the tree graph obtained by cutting every edge from $S$ in a
manner described in section~\ref{sec:two_operations} (see
equation~\eqref{eq:subblock_cut} and around).  For future reference we
list the quadratic forms of the original graph, the graph $\Gcut$ and
the graph $\Gmag$.  We put them in the form that highlights the
similarities between the three.
\begin{align}
  \label{eq:h_clean}
  h(\vx) &= \sum_u q_u x_u^2 - \sum_{(u,v)\in E\setminus S} 2x_u x_v
  - \sum_{(u,v)\in S} 2x_u x_v,\\
  \label{eq:hcut}
  \hcut(\vx) &= \sum_u q_u x_u^2 - \sum_{(u,v)\in E\setminus S} 2x_u x_v
  - \sum_{e_j=(u,v)\in S} \left(-\gamma_j
    x_u^2 - x_v^2 / \gamma_j \right),\\
  \label{eq:hmag}
  \hcut(\vz) &= \sum_u q_u |z_u|^2 
  - \sum_{(u,v)\in E\setminus S} 2\Re\left(\overline{z_u} z_v \right)
  - \sum_{e_j = (u,v)\in S} 2 \Re\left( \overline{z_u} e^{i \alpha_j} z_v \right).
\end{align}
It is essential that the last quadratic form be considered on the
complex vector space $\C^d$ (that can be identified with a real space
of double dimension).  The first two forms can be considered on both
real and complex spaces, with obvious modifications in the latter case.

\subsection{Critical points}
\label{sec:crit_pnts}

Let $f$ be an eigenfunction of the graph $\G$.  We have seen in
Theorem~\ref{thm:cut_mag_interlacing} and its proof that the points
$\alpha=0$ and $\gamma=\hg$ (see equation~\eqref{eq:gamma})
are special: at these points $f$ is an eigenfunction of the graphs
$\Gmaga$ and $\Gcuta$.  Moreover, they are critical points of
the corresponding eigenvalues considered as function of the parameters
$\alpha$ and $\gamma$ respectively.  The result of this section
generalizes this observation.



\begin{theorem}
  \label{thm:cp}
  Let $f$ be an eigenfunction of $H(\G)$ that corresponds to a simple
  eigenvalue $\lambda = \lambda_n(\G)$.  Assume $f$ is non-zero on
  vertices of the graph $\G$.  For every edge
  $(u_j,v_j)\in S$, $j=1,\ldots, \beta$, let
  \begin{equation}
    \label{eq:special_gammas}
    \hg_j = f_{v_j} / f_{u_j}.
  \end{equation}
  Let $p$ denote the number of negatives among the values
  $\hg_j$,
  \begin{equation*}
    p = \#\{\hg_j < 0, \ j=1,\ldots,\beta\}.
  \end{equation*}
  Then 
  \begin{equation}
    \label{eq:eig_transfer}
    \lambda_n(\G) = \lambda_{\zeros_n-p+1}\left(\Gc_{\hg_1,\ldots,\hg_\beta}\right),
  \end{equation}
  where $\zeros_n$ is the number of sign changes of $f$ with respect to the
  graph $\G$.  Moreover, the point $(\hg_1,\ldots,\hg_\beta)$ is a
  critical point of the function
  $\lambda_{\zeros_n-p+1}\left(\Gcut\right)$.

  Similarly for $\Gmag$, 
  \begin{equation}
    \label{eq:eig_transfer_mag}
    \lambda_n(\G) = \lambda_{n}\left(\Gm^{0,\ldots,0}\right)
  \end{equation}
  and $(0,\ldots,0)$ is a critical point of the function $\lambda_n(\Gmag)$.
\end{theorem}

\begin{proof}
  It can be verified directly that $f$ is an eigenfunction of the
  graph $\Gc_{\hg_1,\ldots,\hg_\beta}$.  The nodal bound
  \eqref{eq:nodal_bound_mu} with $\beta=0$ (proved by Fiedler in
  \cite{Fie_cmj75a}, see also \cite{Ber_cmp08}) shows that the
  eigenvalue corresponding to the function $f$ has number $\mu'+1$ in
  the spectrum of the tree $\Gc_{\hg_1,\ldots,\hg_\beta}$.  Here $\mu'$
  is the number of sign changes of $f$ with respect to the tree.  In general,
  this number is different from $\zeros_n$ because we might have cut some
  of the edges on which $f$ was changing sign.  However, according to
  (\ref{eq:special_gammas}), these edges gave rise to \emph{negative}
  values of $\hg_j$, therefore $\mu'= \zeros_n-p$, proving
  equation~\eqref{eq:eig_transfer}.  Equation
  \eqref{eq:eig_transfer_mag} is trivial since $\Gm^{0,\ldots,0} = \G$.

  To prove criticality of the points, we calculate the derivatives.
  Because the corresponding eigenvalue is simple (for the tree, this
  follows from a theorem of Fiedler \cite{Fie_cmj75a} and the fact
  that $f$ does not vanish on vertices), the corresponding eigenvalues
  are analytic functions of the parameters and can be differentiated.
  
  Derivative of $\lambda_{\zeros_n-p+1}\left(\Gcut\right)$ with respect to
  $\gamma_j$ has been calculated in equation~\eqref{eq:deriv_at_max},
  resulting in
  \begin{equation}
    \label{eq:deriv1_cut}
    \left. \frac{\partial}{\partial \gamma_j}
      \lambda_{\zeros_n-p+1}\left(\Gcut\right) \right|_{(\hg_1,\ldots,\hg_\beta)}
    = -|f_{u_j}|^2 + |f_{v_j}|^2 / \hg_j^2 = 0,
  \end{equation}
  where we used the definition of $\hg$,
  equation~\eqref{eq:special_gammas}.

  The derivative of $\lambda_n\left(\Gmag\right)$ can be evaluated
  similarly using \eqref{eq:pert_first_deriv}, leading to
  \begin{equation}
    \label{eq:deriv1_mag}
    \left. \frac{\partial}{\partial \alpha_j} 
      \lambda_{n}\left(\Gmag\right) \right|_{(0,\ldots,0)}
    = -i \overline{f_{u_j}} f_{v_j} + i f_{u_j} \overline{f_{v_j}}
    = \Imag\left(\overline{f_{u_j}} f_{v_j}\right) = 0,
  \end{equation}
  since the eigenfunction $f$ is real-valued.
\end{proof}

\subsection{Index of the eigenvalue on the tree}

In this section we elaborate on the first part of the result of
Theorem~\ref{thm:cp}, namely that $(\hg_1,\ldots,\hg_\beta)$ is a
critical point of the function $\lambda_{\zeros_n-1}\left(\Gcut\right)$.

\begin{theorem}
  \label{thm:inertia_cut_tree}
  Let $f$ be an eigenfunction of $H(\G)$ that corresponds to a simple
  eigenvalue $\lambda=\lambda_n(\G)$.  Assume $f$ is non-zero on
  vertices of the graph $\G$ and has $\zeros_n$ sign changes.  For every edge
  $(u_j,v_j)\in S$, $j=1,\ldots, \beta$, let
  \begin{equation}
    \label{eq:special_gammas_2}
    \hg_j = f_{v} / f_{u}.
  \end{equation}
  As before, $p$ denotes the number of negatives among the values
  $\hg_j$.  Then the point $(\hg_1,\ldots,\hg_\beta)$ as a critical
  point of the function $\lambda_{\zeros_n-p+1}\left(\Gcut\right)$ has
  index $n-1+\beta-\zeros_n$.
\end{theorem}

\begin{proof}
  Denote by $d$ the number of vertices of the graph $\G$.  Consider
  $\hcut(\vx)$, which is the quadratic form on the Hamiltonian of
  $\Gcut$, as a function of $d+\beta$ real variables
  $(x_1,\ldots,x_d,\gamma_1,\ldots,\gamma_\beta)$ on the manifold
  $x_1^1+\ldots+x_{d}^2=1$.  We note that the point $(f_1,\ldots,
  f_d,\hg_1,\ldots,\hg_\beta)$ is a critical point of $\hcut(\vx)$, as
  can be easily shown by explicit computation.  Indeed, the value of
  the Lagrange multiplier is the eigenvalue $\lambda_n$ and the
  gradient of
  \begin{equation*}
    F(x_1,\ldots,x_d, \gamma_1,\ldots \gamma_\beta) = \hcut(\vx) -
    \lambda_n (x_1^1+\ldots+x_{d}^2)
  \end{equation*}
  is zero: the first $d$ equations become the eigenvalue condition $H
  f = \lambda_n f$ and the last $\beta$ are the same as
  \eqref{eq:deriv1_cut}.

  Denote the index of the point $(f_1,\ldots,
  f_d,\hg_1,\ldots,\hg_\beta)$ by $M$.  For every value of
  $(\gamma_1,\ldots,\gamma_\beta)$ locally around the point
  $(\hg_1,\ldots,\hg_\beta)$ the $(\zeros_n-p+1)$-th eigenvector $\fcut$
  of $\Gcut$ is a critical point of $\hcut(\vx)$ as a function of
  $\vx$.  According to Theorem~\ref{thm:index_of_eigenvector}, the
  critical point has index $\zeros_n-p$.

  According to standard perturbation theory (see, e.g.,
  \cite{Kato_book}) the eigenvector $\fcut$ is a smooth (indeed,
  analytic) function of $(\gamma_1,\ldots,\gamma_\beta)$.  This allows
  us to use Theorem~\ref{thm:reduction} via
  Remark~\ref{rem:crit_manifold}, concluding that the critical point
  $(\hg_1,\ldots,\hg_\beta)$ of $\hcut(\fcut)$ has index
  $M-(\zeros_n-p)$.  At this point we observe (see
  Remark~\ref{rem:value_at_crit}) that
  \begin{equation*}
    \hcut(\fcut) = \lambda_{\zeros_n-p+1}\left(\Gcut\right),
  \end{equation*}
  which is the function whose index we strive to evaluate.

  On the other hand, consider $\vx$ varying locally around the point
  $f$, so that the elements of $\vx$ remain bounded away from zero.
  For each fixed $\vx$ we look for a critical point with respect to
  the variables $(\gamma_1,\ldots,\gamma_\beta)$.  The terms of
  $\hcut(\vx)$ that depend on a given $\gamma$ have the form
  \begin{equation}
    \label{eq:gamma_term}
     T(\gamma) = -\gamma x_u^2 - x_v^2 / \gamma.
  \end{equation}
  The critical point is $g = g(\vx)=x_v / x_u$, which is a smooth
  function of $\vx$.  The points $(x_1,\ldots,x_d, g_1,\ldots
  g_\beta)$ define another critical manifold to apply
  Theorem~\ref{thm:reduction} to.  Note that the critical manifold
  includes the point $(f_1,\ldots, f_d,\hg_1,\ldots,\hg_\beta)$.
  Moreover, the critical point with respect to a given $\gamma$ is a
  maximum if $g>0$ and a minimum if $g<0$.  Each point is
  nondegenerate and, moreover, the sign of $g_j$ is locally the same
  as the sign of $\hg_j$ for all $j$.  Different variables $\gamma$
  are not coupled, thus the Hessian is diagonal.  Therefore, the
  inertia of the points on the critical manifold is $(\beta-p, 0, p)$
  --- it is a minimum with respect to $p$ variables and maximum with
  respect to $\beta-p$.

  \begin{figure}[t]
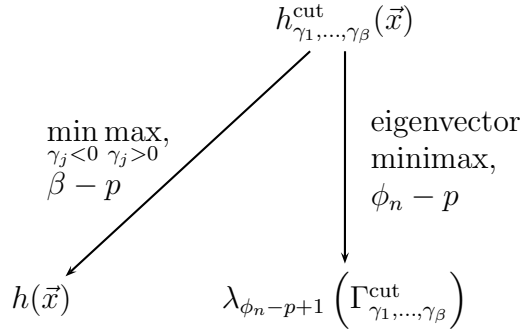

    \label{fig:inertia_cut_tree}
   \begin{psmatrix}[rowsep=3cm,colsep=2cm]
      & $\hcut(\vx)$ \\
      $h(\vx)$ & $\lambda_{\zeros_n-p+1}\left(\Gcut\right)$
      \psset{arrows=->,nodesep=3pt}
      \ncline{1,2}{2,1}<{
        \begin{minipage}{1.0cm}
          $\displaystyle{\min_{\gamma_j<0}\max_{\gamma_j>0}}$,\\
          $\beta-p$
        \end{minipage}
      }
      \ncline{1,2}{2,2}>{
        \begin{minipage}{2.0cm}
          eigenvector\\ 
          minimax,\\
          $\zeros_n-p$
        \end{minipage}
      }
    \end{psmatrix}
   \caption{Schematic diagram of the proof of
      Theorem~\ref{thm:inertia_cut_tree}.  The reductions are
      indicated by arrows, with the description of the parameters that
      are being reduced and the index of the reduction.  Since we know
      the index of the critical point of $h(\vx)$, we can follow the
      diagram, applying the Reduction Theorem, to calculate the index
      of $\lambda_{\zeros_n-p+1}\left(\Gcut\right)$.}
  \end{figure}

  Consider now the function $\hcut(\vx)$ on the critical manifold.
  When $\gamma = g$, the term \eqref{eq:gamma_term} evaluates to
  \begin{equation*}
    T(g) = -2x_ux_v
  \end{equation*}
  and we find that, on the critical manifold, the function
  $\hcut(\vx)$ coincides with the quadratic form of the original
  graph, $h(\vx)$.  The point $\vx=f$, being the $n$-th eigenfunction
  of the graph is a critical point of $h(\vx)$ and has index $n-1$.
  Applying Theorem~\ref{thm:reduction} we obtain
  \begin{equation*}
    n-1 = M - (\beta-p).
  \end{equation*}

  Coming back to $(\hg_1,\ldots,\hg_\beta)$ as a critical point of
  $\lambda_{\zeros_n-p+1}\left(\Gcut\right)$ we conclude that its index
  is
  \begin{equation*}
    M-(\zeros_n-p) = (n-1) + (\beta-p) - (\zeros_n-p) = n-1 + \beta - \zeros_n.
  \end{equation*}
  The steps of the proof are summarized in
  Fig.~\ref{fig:inertia_cut_tree}.
\end{proof}

\begin{remark}
  In \cite{BerRazSmi_prep11} the eigenvalue of the tree graph $\Gcut$
  was interpreted as the energy of the ``partition'' with the given
  number of domains.  Theorem~\ref{thm:inertia_cut_tree} gives another
  route for the proof of the results of \cite{BerRazSmi_prep11}.
\end{remark}

\subsection{Index of the eigenvalue as a function of the magnetic
  field}

Now we move from the critical point on the the tree to the critical
point of the eigenvalue of the graph with magnetic phases.  We
apply the same method, practically retracing our steps, but now the
quadratic form is a function of complex variables.

\begin{theorem}
  \label{thm:inertia_mag}
  Let $f$ be an eigenfunction of $H(\G)$ that corresponds to a simple
  eigenvalue $\lambda=\lambda_n(\G)$.  Assume $f$ is non-zero on
  vertices of the graph $\G$ and has $\zeros_n$ sign changes.  Let $\Gmag$ be
  the graphs with the magnetic phases $\alpha_1,\ldots,\alpha_\beta$
  introduced on the edges from the set $S$.  Then the index of
  $(0,\ldots,0)$ as a critical point of the function
  $\lambda_n(\Gmag)$ is the nodal surplus $\zeros_n-(n-1)$.
\end{theorem}

\begin{proof}
  Let $z$ be the $d$-dimensional vector of complex numbers
  $z_j=x_j+iy_j$.  We consider $\hcut(z)$, the quadratic form of the
  Hamiltonian of $\Gcut$, as a function of $2d+\beta$ real variables
  \begin{equation*}
      x_1,\ldots,x_d,\ 
      y_1,\ldots,y_d,\ 
      \gamma_1,\ldots,\gamma_\beta
  \end{equation*}
  on the manifold $|z_1|^2+\ldots+|z_d|^2=1$.  We also consider
  $\hmag(z)$, the quadratic form of the Hamiltonian of $\Gmag$, as a
  function of $2d+\beta$ real variables on the same manifold.
  See equations~\eqref{eq:hcut} and \eqref{eq:hmag} for explicit
  formulas.

  \begin{figure}[t]
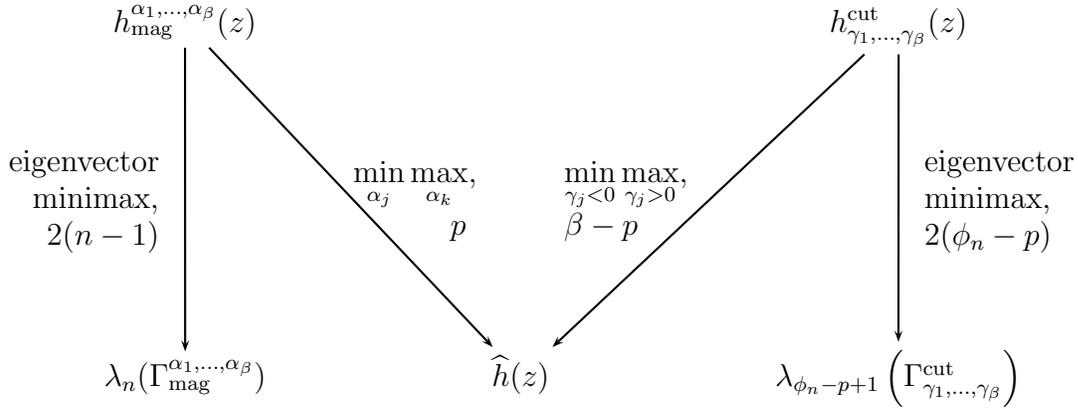

    \label{fig:inertia_mag}
   \begin{center}
   \begin{psmatrix}[rowsep=4cm,colsep=3cm]
      $\hmag(z)$ & & $\hcut(z)$ \\
      $\lambda_n(\Gmag)$ & $\widehat{h}(z)$ 
      & $\lambda_{\zeros_n-p+1}\left(\Gcut\right)$
      \psset{arrows=->,nodesep=3pt}
      \ncline{1,1}{2,1}<{        
        \begin{minipage}[r]{2cm}
          \raggedleft
          eigenvector\\ 
          minimax,\\
          $2(n-1)$
        \end{minipage}
      }
      \ncline{1,1}{2,2}>{
        \hspace{-1cm}
        \begin{minipage}[r]{1.5cm}
          \raggedleft
          $\displaystyle\min_{\alpha_j}\max_{\alpha_k}$,\\
          $p$\vphantom{$\beta$}
        \end{minipage}
      }
      \ncline{1,3}{2,2}<{
        \begin{minipage}{1.0cm}
          $\displaystyle{\min_{\gamma_j<0}\max_{\gamma_j>0}}$,\\
          $\beta-p$
        \end{minipage}
      }
      \ncline{1,3}{2,3}>{
        \begin{minipage}{2.0cm}
          eigenvector\\ 
          minimax,\\
          $2(\zeros_n-p)$
        \end{minipage}
      }
    \end{psmatrix}
    \end{center}
    \caption{Schematic diagram of the proof of
      Theorem~\ref{thm:inertia_mag}.  From
      Theorem~\ref{thm:inertia_cut_tree} we know the index of the
      critical point of $\lambda_{\zeros_n-p+1}\left(\Gcut\right)$
      (bottom right corner).  We then apply the Reduction Theorem four
      times to get the index of $\lambda_n(\Gmag)$ (bottom left
      corner).}
  \end{figure}

  As before, the point $z_j = f_j$, $\gamma_k = \hg_k$ is a critical
  point of the function $\hcut(z)$.  Similarly, $z_j = f_j$,
  $\alpha_k=0$ is a critical point of the function $\hmag(z)$.

  For a fixed vector $z$ in the vicinity of the eigenfunction $f$ we
  find the critical points of functions $\hcut(z)$ and $\hmag(z)$ with
  respect to the corresponding set of parameters.  We choose the
  critical points close to $(\hg_1,\ldots,\hg_\beta)$ and
  $(0,\ldots,0)$ correspondingly.

  Concentrating first on $\hcut(z)$ we observe that now the terms that
  depend on a $\gamma$ have the form
  \begin{equation}
    \label{eq:gamma_term_complex}
        T(\gamma) = -\gamma |z_u|^2 - |z_v|^2 / \gamma.
  \end{equation}
  The critical point is $g(z) = |z_v|/|z_u|$ and it is a maximum or
  minimum depending whether $\gamma>0$ or $\gamma<0$ (the sign of
  $\gamma$ is determined by the sign of the corresponding $\hg$).
  Similarly to the proof of Theorem~\ref{thm:inertia_cut_tree}, the
  index of the critical points that are close to
  $(\hg_1,\ldots,\hg_\beta)$ is $\beta-p$, where $p$ is the number of
  negatives among $\hg_j$.  On the critical manifold the
  term~\eqref{eq:gamma_term_complex} evaluates to
  \begin{equation}
    \label{eq:gamma_term_evaluated}
    T(g) = -2\sign(\hg) |z_u| |z_v|.
  \end{equation}
  
  Moving on to $\hmag(z)$ we first investigate its dependence on a
  single phase $\alpha$.  The terms involving $\alpha$ are of the form
  \begin{equation}
    \label{eq:term_alpha_complex}
    T_m(\alpha) = -\overline{z_u} e^{i\alpha} z_2 - \overline{z_v}
    e^{-i \alpha} z_u = -2|z_u| |z_v| \cos(\alpha - \theta_{uv}),
  \end{equation}
  where $\theta_{uv} = \arg(z_u/z_v)$.  Since $z_u, z_v \in \C$ are
  close to $f_u, f_v \in \R$ correspondingly, the angle $\theta_{uv}$
  is close to $0$ if $\hg = f_v/f_u > 0$ and close to $\pi$ if $\hg =
  f_v/f_u <0$.  The corresponding critical point $a=a(z)$ is then
  minimum ($\hg>0$) or maximum ($\hg<0$).  

  Considering now the critical point of $\hmag(z)$ as a function of
  all parameters $\alpha_1,\ldots,\alpha_\beta$ (the variable $z$ is
  fixed), we find that the Hessian is diagonal and therefore the index
  is the number of coordinate-wise maxima, that is $p$.  Moreover, on
  the critical manifold the term \eqref{eq:term_alpha_complex} takes
  the form identical to equation~\eqref{eq:gamma_term_evaluated}.
  Comparing the forms $\hcut(z)$ and $\hmag(z)$ we conclude that, for
  each value of $z$, they coincide at their respective critical
  points, assuming the value
  \begin{equation*}
    \widehat{h}(z) = \sum_u q_u |z_u|^2 
    - \sum_{(u,v)\in E\setminus S} 2\Re\left(\overline{z_u} z_v \right)
    - \sum_{e_j=(u,v) \in S} 2\sign(\hg_j) |z_u| |z_v|.
  \end{equation*}
  Thus indices of the point $z=f$ with respect to the corresponding
  critical manifold of $\hcut(z)$ and $\hmag(z)$ are equal.  This
  establishes a bridge between the two forms.  Note that the bridging
  function $\widehat{h}(z)$ is different from the form $h(z)$ of the
  original graph $\G$.  The relations between the quadratic forms and
  the eigenvalues as functions of parameters are sketched in
  Fig.~\ref{fig:inertia_mag}.  We will now be repeatedly applying the
  reduction of critical manifolds, Theorem~\ref{thm:reduction}.  The
  indices of the reductions are noted along the arrows in
  Fig.~\ref{fig:inertia_mag} (see below for explanations).
  
  Fixing the parameters $\gamma_1,\ldots,\gamma_\beta$, we find the
  critical point of $\hcut(z)$ as a function of $z$.  Locally around
  the point $z_j = f_j$, $\gamma_k = \hg_k$, the critical point is the
  $(\zeros_n-p+1)$-th eigenvector of the graph $\Gcut$.  According to
  Theorem~\ref{thm_index_of_eigenvector_complex} its index is
  $2(\zeros_n-p)$.  We apply the reduction theorem to the result of
  Theorem~\ref{thm:inertia_cut_tree} to conclude that the index of the
  point $z = f$, $\gamma = \hg$ with respect to all parameters is
  \begin{equation*}
    (n-1+\beta-\zeros_n) + 2(\zeros_n-p) = n-1+\beta+\zeros_n - 2p.
  \end{equation*}
  Now we go down to the ``bridge'' critical point and then back up to
  the critical point $z=f$, $\alpha=0$ of the function $\hmag(z)$, applying the
  reduction theorem both ways.  For the latter critical point we
  obtain the index
  \begin{equation*}
    (n-1+\beta+\zeros_n - 2p) - (\beta-p) + p = n-1 + \zeros_n.
  \end{equation*}

  Finally, for every choice of parameters
  $\alpha_1,\ldots,\alpha_\beta$, we find the critical point with
  respect to $z$ that is close to $z=f$.  This critical point $\fmag$ is the
  $n$-th eigenvector of the graph $\Gmag$ (since $f$ is the $n$-th
  eigenvector of $\Gm^{0,\ldots,0}$).  According to
  Theorem~\ref{thm_index_of_eigenvector_complex} it has index
  $2(n-1)$.  We apply the reduction theorem one last time to conclude
  that the point $(0,\ldots,0)$ as a critical point of $\hmag(\fmag)$
  has index
  \begin{equation*}
    n-1 + \zeros_n - 2(n-1) = \zeros_n - (n-1).
  \end{equation*}
  Since $\hmag(\fmag) = \lambda_n(\Gmag)$, this concludes the proof of
  the theorem.
\end{proof}

\section{Discussion}
\label{sec:discuss}

Perhaps the most important feature of Theorem~\ref{thm:main} is that
it allows to access some of the features of the eigenfunction via the
behavior of the corresponding eigenvalue under perturbation.  It is
known that the eigenvalues of the Laplacian are connected to the
statistics of the closed paths on the graph.  The connection is given
through the so-called ``trace formulae'', which can be obtained from a
graph analogue of the Selberg zeta function, the Ihara zeta function
\cite{Iha_jmsj66,Bas_ijm92,StaTer_am96}.  An extension by Bartholdi
\cite{Bar_em99} (see also \cite{MizSat_jac05}) was used in
\cite{OreGodSmi_jpa09} to obtain a family of trace formulae including
the ones for the magnetic Laplacian.  Thus, the closed paths on the
graph determine the spectrum of the magnetic Laplacian which, in turn,
determines the nodal count.  This, in principle, establishes the
existence of a general connection between the nodal count and the
closed paths.  However, we are not aware of any concrete general
formulas.  We note that such a connection has been earlier conjectured
by Smilansky, with special cases reported in
\cite{GnuKarSmi_prl06,AroSmi_prep10}.

We note that the eigenvalue $\lambda_n(\Gmag)$ featuring in this paper
is a well-studied object.  It is the dispersion relation for the
maximal Abelian cover of the graph $\G$.  One of the interesting
questions regarding this object is the ``full spectrum property''
\cite{HigShi_cm04,HigNom_ejc09,Sun_pspm08}: whether the continuous
spectrum of the cover graph of a regular graph --- in our terms, the
union of ranges of the functions $\lambda_n(\Gmag)$ --- contains no
gaps.  This question can be reformulated in terms of eigenfunctions of
graphs $\Gmag$ with all $\alpha_j=0$ or $\pi$ that have minimal and
maximal number of sign changes.

This, in turn, is related to the question of whether the extrema of
the dispersion relation are always achieved at the symmetry points
(namely, all $\alpha_j=0$ or $\pi$).  Examples to the contrary have
been put forward in \cite{HarKucSob_jpa07,ExnKucWin_jpa10}.  However,
an important question remains, how can one characterize the extremal
points that are not points of symmetry?  In this direction, the
duality with the cut graphs (Section~\ref{sec:duality}) might provide
some answers.  One can speculate that critical points of dispersion
relation correspond to critical points of the eigenvalues of the cut
graph $\Gcut$ that \emph{do not} give rise to the eigenfunction of the
graph $\G$.  Further, we conjecture that these ``unclaimed'' critical
points correspond to eigenfunctions of $\G$ modified by enforcing
Dirichlet conditions at some vertices.

The results of the present paper are derived under the assumption that
the eigenvalue is non-degenerate.  While this is the generic situation
with respect to the change in the potential $Q$, it is also
interesting to consider what happens in the degenerate case.  Linear
Zeeman effect (the magnetic perturbation splits eigenvalues) suggests
that the singularities of $\lambda_n(\Gmag)$ are conical.  It should
be possible to define the index of the singularity point that does not
rely on differentiability.

Finally, it would be most interesting to generalize the results of the
present paper to domains of $\R^d$.  However, we immediately encounter
a conceptual problem --- the ``number'' of zeros is infinite.  Still,
some measure of instability of the eigenvalue under magnetic
perturbation should be related to some measure of the zero set of the
corresponding eigenfunction.  This can be intuitively visualized by
approximating the domain eigenfunction by eigenfunctions of a discrete
mesh.

\section{Acknowledgment}
\label{sec:ack}

The impetus for considering magnetic perturbation was given to the
author (perhaps unwittingly) by two people: B.~Helffer and
P.~Kuchment.  B.~Helffer, in a talk at ``Selected topics in spectral
theory'' at Erwin Schr\"odinger Institute (Vienna), Jan 13-15, 2011,
described the conjecture that the minimal spectral partition of a
domain can be obtained by finding the minimal energy of the Laplacian
with a number of Aharonov-Bohm flux lines (minimization with respect
to the number of the flux lines as well as their position).
P.~Kuchment gave the author a preprint by
Rueckriemen\cite{Ruc_prep11}, which discussed objects similar to the
ones introduced in a study of minimal spectral partitions on graphs
\cite{Banetal_cmpXX}.  We are also grateful to P.~Kuchment for many
subsequent enlightening discussions, in particular regarding topics
touched upon in Section~\ref{sec:discuss}.  The visit to ESI in Vienna
has been possible through generosity of the Institute.  The author is
supported by the NSF grant DMS-0907968.

\bibliographystyle{abbrv}
\bibliography{nodal_mag}

\begin{thebibliography}{10}

\bibitem{AroSmi_prep10}
A.~Aronovitch and U.~Smilansky.
\newblock Trace formula for counting nodal domains on the boundaries of chaotic
  2{D} billiards.
\newblock preprint {\tt arXiv:1006.5656 [quant-ph]}, 2010.

\bibitem{Banetal_cmpXX}
R.~Band, G.~Berkolaiko, H.~Raz, and U.~Smilansky.
\newblock On the connection between the number of nodal domains on quantum
  graphs and the stability of graph partitions.
\newblock {\em Accepted at Comm. Math. Phys.}, 2011.
\newblock preprint arXiv:1103.1423 [math-ph].

\bibitem{BanOreSmi_pspm08}
R.~Band, I.~Oren, and U.~Smilansky.
\newblock Nodal domains on graphs---how to count them and why?
\newblock In {\em Analysis on graphs and its applications}, volume~77 of {\em
  Proc. Sympos. Pure Math.}, pages 5--27. Amer. Math. Soc., Providence, RI,
  2008.

\bibitem{Bar_em99}
L.~Bartholdi.
\newblock Counting paths in graphs.
\newblock {\em Enseign. Math. (2)}, 45(1-2):83--131, 1999.

\bibitem{Bas_ijm92}
H.~Bass.
\newblock The {I}hara-{S}elberg zeta function of a tree lattice.
\newblock {\em Internat. J. Math.}, 3(6):717--797, 1992.

\bibitem{Ber_cmp08}
G.~Berkolaiko.
\newblock A lower bound for nodal count on discrete and metric graphs.
\newblock {\em Comm. Math. Phys.}, 278(3):803--819, 2008.

\bibitem{BerKucSmi_prep11}
G.~Berkolaiko, P.~Kuchment, and U.~Smilansky.
\newblock Critical partitions and nodal deficiency of billiard eigenfunctions.
\newblock preprint arXiv:1107.3489 [math-ph], 2011.

\bibitem{BerRazSmi_prep11}
G.~Berkolaiko, H.~Raz, and U.~Smilansky.
\newblock Stability of nodal structures in graph eigenfunctions and its
  relation to the nodal domain count.
\newblock preprint arXiv:1110.3802 [math-ph], 2011.

\bibitem{Hooke}
T.~Birch.
\newblock {\em History of the Royal Society of London}.
\newblock A. Millar, London, 1756.

\bibitem{Biy_laa03}
T.~B{\i}y{\i}ko{\u{g}}lu.
\newblock A discrete nodal domain theorem for trees.
\newblock {\em Linear Algebra Appl.}, 360:197--205, 2003.

\bibitem{BiyLeySta_book}
T.~B{\i}y{\i}ko{\u{g}}lu, J.~Leydold, and P.~F. Stadler.
\newblock {\em Laplacian eigenvectors of graphs}, volume 1915 of {\em Lecture
  Notes in Mathematics}.
\newblock Springer, Berlin, 2007.
\newblock Perron-Frobenius and Faber-Krahn type theorems.

\bibitem{BluGnuSmi_prl02}
G.~Blum, S.~Gnutzmann, and U.~Smilansky.
\newblock Nodal domains statistics: A criterion for quantum chaos.
\newblock {\em Phys. Rev. Lett.}, 88:114101, Mar 2002.

\bibitem{BogSch_prl02}
E.~Bogomolny and C.~Schmit.
\newblock Percolation model for nodal domains of chaotic wave functions.
\newblock {\em Phys. Rev. Lett.}, 88:114102, Mar 2002.

\bibitem{Chladni}
E.~Chladni.
\newblock {\em Entdeckungen \"uber die {T}heorie des {K}langes}.
\newblock Weidmanns Erben und Reich, Leipzig, 1787.

\bibitem{CdV_spectre}
Y.~Colin~de Verdi{\`e}re.
\newblock {\em Spectres de graphes}, volume~4 of {\em Cours Sp\'ecialis\'es
  [Specialized Courses]}.
\newblock Soci\'et\'e Math\'ematique de France, Paris, 1998.

\bibitem{CdVToHTru_prep10}
Y.~Colin De~Verdière, N.~Torki-Hamza, and F.~Truc.
\newblock Essential self-adjointness for combinatorial {S}chr\"odinger
  operators {III} --- {M}agnetic fields.
\newblock arXiv:1011.6492 [math.SP].

\bibitem{Cou_ngwgmp23}
R.~Courant.
\newblock Ein allgemeiner {S}atz zur {T}heorie der {E}igenfuktionen
  selbstadjungierter {D}ifferentialausdr\"ucke.
\newblock {\em Nachr. Ges. Wiss. G\"ottingen Math Phys}, pages 81--84, 1923.

\bibitem{CourantHilbert_volume1}
R.~Courant and D.~Hilbert.
\newblock {\em Methods of mathematical physics. {V}ol. {I}}.
\newblock Interscience Publishers, Inc., New York, N.Y., 1953.

\bibitem{DavGlaLeySta_laa01}
E.~B. Davies, G.~M.~L. Gladwell, J.~Leydold, and P.~F. Stadler.
\newblock Discrete nodal domain theorems.
\newblock {\em Linear Algebra Appl.}, 336:51--60, 2001.

\bibitem{Elo_jpa08}
Y.~Elon.
\newblock Eigenvectors of the discrete {L}aplacian on regular graphs---a
  statistical approach.
\newblock {\em J. Phys. A}, 41(43):435203, 17, 2008.

\bibitem{ExnKucWin_jpa10}
P.~Exner, P.~Kuchment, and B.~Winn.
\newblock On the location of spectral edges in {$\Bbb Z$}-periodic media.
\newblock {\em J. Phys. A}, 43(47):474022, 8, 2010.

\bibitem{Fie_cmj75a}
M.~Fiedler.
\newblock Eigenvectors of acyclic matrices.
\newblock {\em Czechoslovak Math. J.}, 25(100)(4):607--618, 1975.

\bibitem{Galileo}
Galileo.
\newblock {\em Discorsi e dimostrazioni matematiche, intorno \`a due nuove
  scienze}.
\newblock Louis Elsevier, Leiden, 1638.

\bibitem{GnuKarSmi_prl06}
S.~Gnutzmann, P.~D. Karageorge, and U.~Smilansky.
\newblock Can one count the shape of a drum?
\newblock {\em Phys. Rev. Lett.}, 97(9):090201, 2006.

\bibitem{Har_ppsa55}
P.~G. Harper.
\newblock Single band motion of conduction electrons in a uniform magnetic
  field.
\newblock {\em Proc. Phys. Soc. Lon. A}, 68:874--878, 1955.

\bibitem{HarKucSob_jpa07}
J.~M. Harrison, P.~Kuchment, A.~Sobolev, and B.~Winn.
\newblock On occurrence of spectral edges for periodic operators inside the
  {B}rillouin zone.
\newblock {\em J. Phys. A}, 40(27):7597--7618, 2007.

\bibitem{HelHofTer_aihp09}
B.~Helffer, T.~Hoffmann-Ostenhof, and S.~Terracini.
\newblock Nodal domains and spectral minimal partitions.
\newblock {\em Ann. Inst. H. Poincar\'e Anal. Non Lin\'eaire}, 26(1):101--138,
  2009.

\bibitem{HigNom_ejc09}
Y.~Higuchi and Y.~Nomura.
\newblock Spectral structure of the {L}aplacian on a covering graph.
\newblock {\em European J. Combin.}, 30(2):570--585, 2009.

\bibitem{HigShi_cm04}
Y.~Higuchi and T.~Shirai.
\newblock Some spectral and geometric properties for infinite graphs.
\newblock In {\em Discrete geometric analysis}, volume 347 of {\em Contemp.
  Math.}, pages 29--56. Amer. Math. Soc., Providence, RI, 2004.

\bibitem{Hof_prb76}
D.~R. Hofstadter.
\newblock Energy levels and wave functions of bloch electrons in rational and
  irrational magnetic fields.
\newblock {\em Phys. Rev. B}, 14:2239--2249, Sep 1976.

\bibitem{HornJohnson_matrix}
R.~A. Horn and C.~R. Johnson.
\newblock {\em Matrix analysis}.
\newblock Cambridge University Press, Cambridge, 1985.

\bibitem{Iha_jmsj66}
Y.~Ihara.
\newblock On discrete subgroups of the two by two projective linear group over
  {$p$}-adic fields.
\newblock {\em J. Math. Soc. Japan}, 18:219--235, 1966.

\bibitem{Kato_book}
T.~Kato.
\newblock {\em Perturbation theory for linear operators}.
\newblock Springer-Verlag, Berlin, second edition, 1976.
\newblock Grundlehren der Mathematischen Wissenschaften, Band 132.

\bibitem{LieLos_dmj93}
E.~H. Lieb and M.~Loss.
\newblock Fluxes, {L}aplacians, and {K}asteleyn's theorem.
\newblock {\em Duke Math. J.}, 71(2):337--363, 1993.

\bibitem{Leonardo}
E.~MacCurdy.
\newblock {\em Notebooks of Leonardo da Vinci}.
\newblock Jonathtan Cape, London, 1938.

\bibitem{MizSat_jac05}
H.~Mizuno and I.~Sato.
\newblock A new proof of {B}artholdi's theorem.
\newblock {\em J. Algebraic Combin.}, 22(3):259--271, 2005.

\bibitem{NazSod_ajm09}
F.~Nazarov and M.~Sodin.
\newblock On the number of nodal domains of random spherical harmonics.
\newblock {\em Amer. J. Math.}, 131(5):1337--1357, 2009.

\bibitem{Ore_jpa07}
I.~Oren.
\newblock Nodal domain counts and the chromatic number of graphs.
\newblock {\em J. Phys. A: Math. Theor.}, 40:9825--9832, 2007.

\bibitem{OreGodSmi_jpa09}
I.~Oren, A.~Godel, and U.~Smilansky.
\newblock Trace formulae and spectral statistics for discrete {L}aplacians on
  regular graphs. {I}.
\newblock {\em J. Phys. A}, 42(41):415101, 20, 2009.

\bibitem{Ruc_prep11}
R.~Rueckriemen.
\newblock Recovering quantum graphs from their {B}loch spectrum.
\newblock arXiv:1101.6002v2 [math.SP], 2011.

\bibitem{Shu_cmp94}
M.~Shubin.
\newblock Discrete magnetic {L}aplacian.
\newblock {\em Comm. Math. Phys.}, 164:259--275, 1994.

\bibitem{SS07}
U.~Smilansky and H.-J. St\"ockmann., editors.
\newblock {\em Nodal Patterns in Physics and Mathematics}, volume 145 of {\em
  The European Physical Journal - Special Topics}.
\newblock Springer, 2007.

\bibitem{StaTer_am96}
H.~M. Stark and A.~A. Terras.
\newblock Zeta functions of finite graphs and coverings.
\newblock {\em Adv. Math.}, 121(1):124--165, 1996.

\bibitem{Stu_jmpa36a}
C.~Sturm.
\newblock M\'emoire sur une classe d'\'equations \`a diff\'erences partielles.
\newblock {\em J. Math. Pures Appl.}, 1:373--444, 1836.

\bibitem{Sun_conf93}
T.~Sunada.
\newblock Generalized {H}arper operator on a graph.
\newblock Talk at "Zeta Functions in Number Theory and Geometric Analysis",
  John Hopkins University, March 29 -- April 3, 1993.

\bibitem{Sun_cm94}
T.~Sunada.
\newblock A discrete analogue of periodic magnetic {S}chr\"odinger operators.
\newblock In {\em Geometry of the spectrum ({S}eattle, {WA}, 1993)}, volume 173
  of {\em Contemp. Math.}, pages 283--299. Amer. Math. Soc., Providence, RI,
  1994.

\bibitem{Sun_pspm08}
T.~Sunada.
\newblock Discrete geometric analysis.
\newblock In {\em Analysis on graphs and its applications}, volume~77 of {\em
  Proc. Sympos. Pure Math.}, pages 51--83. Amer. Math. Soc., Providence, RI,
  2008.

\end{thebibliography}

\end{document}